\documentclass[runningheads]{llncs}

\usepackage{amssymb}
\usepackage{nicefrac} 
\usepackage{graphicx}
\usepackage{csquotes}

\usepackage{amsmath}
\usepackage{dsfont}
\usepackage{listings}
\usepackage{color}
\usepackage[numbers]{natbib}
\makeatletter				
\renewcommand\bibsection	
{
  \section*{\refname
    \@mkboth{\MakeUppercase{\refname}}{\MakeUppercase{\refname}}}
}
\makeatother

\usepackage[top=1.1in, bottom=1.1in, left=1.1in, right=1.1in]{geometry}

\newcommand{\sq}{\hbox{\rlap{$\sqcap$}$\sqcup$}}
\renewcommand{\qed}{\hspace*{\fill}\sq}
\renewenvironment{proof}{\noindent {\bf Proof.}\ }{\qed\par\vskip 4mm\par}

\newcommand{\N}{\mbox{$\mathds{N}$}}

\newcommand{\bigO}{\ensuremath{\mathcal{O}}}
\newcommand{\smallO}{\ensuremath{o}}

\newcounter{ctr}
\newcommand{\deftxt}[2]{\def #1 {#2}}

\newcommand{\lem}[2]{
\begin{lemma}\label{#1}
\deftxt{\string \lem:#1txt}{#2}
\string \lem:#1txt
\end{lemma}
}

\renewcommand{\c}{\kappa}
\newcommand{\s}{\mathfrak{s}}

\newcommand{\EQlabel}[1]{\label{#1}} 
\newcommand{\nt}{\notag}
\newcommand{\EQref}[1]{\ref{#1}} 

\newcommand{\I}{\mathcal I}
\newcommand{\U}{\mathcal U}
\newcommand{\R}{\mathcal R}

 \newtheorem{observation}{Observation}

\usepackage{url}

\begin{document}

\mainmatter  

\title{- Awareness and Movement vs. the Spread of Epidemics -\\Analyzing a Dynamic Model for Urban Social/Technological Networks \thanks{eligible for the best student paper award}}

 \titlerunning{The Importance of Movement and Awareness}

%
%
\author{Robert Els\"asser\inst{1}
\and Adrian Ogierman\inst{2}}
%

\institute{University of Paderborn, Institute for Computer Science, F\"urstenallee 11, 33102 Paderborn, Germany, \\Tel.: +49-5251-606692, \email{elsa@upb.de}
\and University of Paderborn, Institute for Computer Science, F\"urstenallee 11, 33102 Paderborn, Germany, \\Tel.: +49-5251-606722, \email{adriano@upb.de}
}

\maketitle

\begin{abstract}
We consider the spread of epidemics in technological and social networks. How do people react? Does awareness and cautious behavior help? We analyze these questions and present a dynamic model to describe the movement of individuals and/or their mobile devices in a certain (idealistic) urban environment. Furthermore, our model incorporates the fact that different locations can accommodate a different number of people (possibly with their mobile devices), who may pass the infection to each other. We obtain two main 
results.
First, we prove that w.r.t. our model at least a small part of the system will remain uninfected even if no countermeasures are taken. The second result shows that with certain counteractions in use, which only influence the individuals' behavior, a  prevalent epidemic can be avoided. The results explain possible courses of a disease, and point out why cost-efficient countermeasures may reduce the number of total infections from a high percentage of the population to a negligible fraction.
\end{abstract}

\newpage
\setcounter{page}{1}
\section{Introduction}\label{intro}
How can I protect myself if an epidemic outbreak occurs? This question 
concerns us all. Individuals protect themselves by avoiding 
contacts to infected people,
companies exploit all possibilities to keep 
their employees viable, and governments try to ensure the public 
health and safety. However, in all these cases one must take 
into account that we live 
in a mobile society, and prohibiting personal contacts between individuals -
which define a so called \textit{social interaction network} - is not desirable.

Since the beginning of time, people inhabited different areas 
and moved with their tribes when needed. Although this behavior 
changed over time, traveling became more and more popular 
for various reasons. Nowadays, even the mobility within a single 
city is extremely high \cite{SubPop09}. Our society relies on delivery systems, 
personal services of different kind, and close co-operation. However, 
mechanisms that are able to control the spread of epidemics (even just within the urban 
area of a large 
city) are quite expensive. Protection systems and an enormous amount of mostly 
expensive antidotes are needed, which may not be available 
when a yet unknown disease appears. Moreover, the immunization of the 
majority of the population is most likely not feasible.

The main question is, how can the society react if an epidemic outbreak occurs? In a 
developed country there is a certain budget for health care. The amount provided by 
this budget is surely finite, but one may assume that it ensures the medication 
of at least a small fraction of the population. Hence, one can ask whether we 
can embank an epidemic by using only our limited resources.

Such a resource may be the use of modern media to change the individuals' 
behavior. One can assume that in a modern civilization, individuals are well 
connected to various types of information sources such as newspaper, television, 
and the Internet. Since these powerful and effective tools are already established 
and often used, a government can easily warn the population. From this time on 
(for the duration of the epidemic) the majority of the informed people 
will act more carefully.


On the technological site, we also have to deal with the problem
of epidemics. Nowadays, plenty of different 
devices provided with microprocessors are guiding us through the day. For many of 
us these became a permanent assistant like smart phones and notebooks. 
Different devices with similar communication interfaces are able 
to form a so called ad-hoc network.
In 
this context Bluetooth plays an important role, and we
refer to the network defined by interactions 
between various devices equipped with this technology as a \textit{technological network}. 
Although Bluetooth is quite useful, it also became more and more attractive to people willing 
to exploit its weaknesses. 
Many different attacks 
were developed and several weaknesses have been exploited since the appearance of this technology 
\cite{LooAlfred09, Dunning10}. 
One of the main problems consists for example in the possibility to crack the PIN needed to 
establish a connection between two devices \cite{ShakedWool05}.
Currently, 
it is even possible to infect e.g. a smart phone even without the active 
assistance of the owner by using malware like worms \cite{Dunning10}. Such a compromised 
device can be used to automatically infect other Bluetooth devices in its current vicinity (and 
in the ad-hoc network to which it belongs). 
Even though in new Bluetooth versions several security issues have been fixed, many devices which are 
in use remain vulnerable. However, the entire replacement of those may take years and result in
huge expenses for consumers.


Motivated by the fact that mobile devices always follow 
the movement pattern of their owners,
we present a model, which describes the dynamic behavior of individuals 
in a certain urban environment. Furthermore, we utilize proper parameters to cover the main 
characteristics described above \cite{EGKM04}. In the rest of this paper, the terminology follows 
the usual notions known from the field of epidemic diseases, and we do not distinguish between 
the spread of epidemics in social interaction and technological networks.
We consider the following models (cf. Section \ref{model_annotation} for a more formal description).

(A) \emph{Emerging country model:} In this model we assume that the medical equipment is rather 
limited. Thus, an epidemic may survive for a long time. Furthermore, we assume that the warning 
possibilities are also very limited. Media like television or the Internet does not exist or provides 
rather small benefit. In other words, this type of model provides only negligible possibilities 
for disease control. The most interesting question is whether it is possible for the population 
to survive despite these circumstances.

(B) \emph{Industrialized country model:} Compared to the emerging country model, we now assume a 
significantly larger budget for disease control. This assumption implies a few facts. First, the 
government is able to cure a large number of individuals simultaneously. Furthermore, the health 
care system is capable of taking people into quarantine to protect the rest of the population, and 
the media is also omnipresent. Thus, 
a message communicated via newspaper, television, or the Internet will be 
received by the whole population within a short time. The main question is 
whether these possibilities provide a positive impact on the embankment of an epidemic.

\subsection{Related Work}\label{related}

One of the most important processes analyzed on social networks (in the usual sense) 
is the spread of diseases. There is plenty of work considering 
epidemiological processes in different scenarios and on various 
networks. In this subsection, we only describe the papers which are 
closely related to our results.

The simplest model of mathematical disease spreading is the so called 
SIR model (see e.g.~\cite{Het00,New03}). The population is divided into three 
categories: susceptible (S), i.e., all individuals which do not have the disease yet
but can become infected, infective (I), i.e., the individuals which have the 
disease and can infect others, and recovered (R), i.e., all individuals which 
recovered and have permanent immunity (or have been removed from the system). 
Most papers model the spread of epidemics using a differential equation based on 
the assumption that any susceptible individual has uniform probability $\beta$ to 
become infected from any infective individual. Furthermore, any infected player 
recovers at some stochastically constant rate $\gamma$. 

This traditional (fully mixed) model can easily be generalized to a 
network. It has been observed that such a case can be modeled by bond percolation
on the underlying graph \cite{Gas83,New02}. Callaway et al.~\cite{CNSW00} considered 
this model on graph classes constructed by the so called configuration model (i.e., 
a random graph with a given degree distribution). 
The SIR model has also been analyzed 
in some other scenarios, including 
various kinds of correlations between the rates of infection or the 
infectivity times, in networks with a more complex structure containing 
different types of vertices \cite{New02}, or in graphs with correlations between 
the degrees of the vertices \cite{MV03}. Interestingly, for certain graphs 
with a power law degree distribution, there is no constant threshold for the
epidemic outbreak as long as the power law exponent is less than $3$ \cite{New03} (which 
is the case in most real world networks, e.g.~\cite{FFF99,AH00,ASBS00,RFI02}). If the network 
is embedded into a low dimensional space, or it has high transitivity, then there 
might exist a non-zero threshold for certain types of correlations between 
vertices. However, none of the papers above considered the dynamic movement 
of individuals, which seems to be the main source of the spread of diseases in urban 
areas \cite{EGKM04}. 

Borgs et al. \cite{Borgs10} focused on how to distribute antidote to control epidemics. The authors analyzed a variant of the contact process in the \textit{susceptible-infected-susceptible (SIS) model} on a finite graph in which the cure rate is allowed to vary from one vertex to the next. That means the rate $\rho_v$ at which an infected node $v$ becomes healthy is proportional to the amount of antidote it received, given a fixed amount of antidote $R=\sum_{x \in V} \rho_x$ for the whole network. The authors studied contact tracing on the star graph and the distribution of the antidote proportional to the node degree on expander graphs and general graphs with bounded average degree as curing mechanisms.
They state that using contact tracing on a star graph would require a total amount of antidote which is super-linear in the number of vertices. Here, contact tracing on a graph means that the cure rate is adjusted to $\rho_v = \rho + \rho' d_v^*(t)$ at every time $t$, where $d_v^*(t)$ denotes the number of infected neighbors of $v$ at time $t$ and $\rho>0$ is some constant. From the point where the number of infected leaves is $\omega(\rho/\beta)$, the epidemic can not be prevented anymore with high probability, even if the amount of antidote is $\beta n^{4/3-\smallO(1)}$, where $\beta$ describes the probability of a node to become infected by a neighbor.
However, setting $\rho_v$ proportional to the degree requires an amount of antidote that scales only linearly with the number of vertices, even on general graphs, provided the average degree is bounded.
On the other side, even if the underlying graph is an expander, then curing proportional to the degree cannot reduce the needed amount of antidote by more than a constant factor.

In \cite{EGKM04}, 
Eubank et al.~modeled physical contact patterns, which result from movement of 
individuals between specific locations, by a dynamic bipartite graph. 
The graph is partitioned into two parts. The first part contains the people
who carry out their daily activities moving between different locations. The 
other part represents the various locations in a certain city. There is an edge 
between two nodes, if the corresponding individual visits a certain location 
at a given time. Certainly, the graph changes dynamically at every time step.

Eubank et al.~\cite{EGKM04, chowell03} analyzed the corresponding network for Portland, 
Oregon. According to their study, the degrees of the nodes 
describing different locations follow a power law distribution with exponent 
around $2.8$\footnote{In \cite{EGKM04} the degree represents the number 
of individuals visiting these places over a time period of 24 hours.}. For many 
epidemics, transmission occurs between individuals being simultaneously 
at the same place, and then people's movement is mainly responsible for 
the spread of the disease.

The authors of \cite{EGKM04} also considered different countermeasures in order 
to avoid an epidemic outbreak. They stated that early detection combined with targeted 
vaccination are effective ways of defense compared to mass vaccination of a population.
However, in most cases it is not possible to find the individuals having many acquaintances,
and in many cases vaccinations cannot even be applied (e.g., SARS, or swine flu in its 
early stages).

In addition to the theoretical papers described above, plenty of simulation work has been done. 
Two of the most popular approaches are the so called agent-based and structured 
meta-population-based, respectively (cf. \cite{VespAgentVsMeta10, Jaffry08}). Both models 
have their advantages and weaknesses. The main idea of the meta-population approach is to model 
whole regions, e.g. georeferenced census areas around airport hubs \cite{VespH1N109}, and connect 
them by a mobility network. Then, within these regions the spread of epidemics is analyzed by 
using the well known mean field theory. In contrast, the agent-based approach models individuals 
with agents in order to simulate their behavior. In this context, the agents may be defined very 
precisely, including e.g. race, gender, educational level, nutritional status, age, priority groups, 
participant class, etc. \cite{Lee08, Lee10}, and thus provide a huge amount of detailed data conditioned on the agents setting. 
Furthermore, these kind of models are able to integrate different locations like schools, theaters 
and so on. Thus, an agent may or may not be infected depending on his own choices and the ones made 
by agents in his vicinity. The main issue about the agent-based approach is the huge amount of 
computational capacity needed to simulate huge cities,  continents or even the world itself 
\cite{VespAgentVsMeta10}. This limitation can be attenuated by reducing the number of agents, 
which then entails a decreasing accuracy of the simulation. In the meta-population approach the 
simulation costs are lower, sacrificing accuracy and some kind of noncollectable data.

To combine the advantages of both systems, hybrid environments were implemented (e.g. \cite{hybrid07}). 
The main idea of such systems is to use an agent-based approach at the beginning of the simulation 
up to some point where a sufficient number of agents are infected. Then, the system switches to a 
meta-population-based approach. Certainly, such a system combines the high accuracy of the 
agent-based simulations at the beginning of the procedure with the faster simulation speed of 
the meta-population-based approach at stages, in which both systems seem to provide similar predictions.

Nonetheless, these kind of simulations confirm the positive impact of 
non-pharmaceutical countermeasures, which is underpinned by examinations 
on real data (e.g. \cite{Markel07}). 
Germann et al.~\cite{Germann06} 
investigated the spread of a pandemic strain of influenza virus through the U.S. population. 
They used publicly available 2000 U.S. Census data to identify seven so-called mixing groups, in 
which each individual may interact with any other member. Each class of mixing group 
is characterized by its own set of age-dependent probabilities for person-to-person transmission 
of the disease.   
They considered different combinations of socially targeted antiviral prophylaxis, dynamic mass 
vaccination, closure of schools and social distancing as countermeasures in use, and simulated 
them with different basic reproductive numbers $R_0$. It turned out that specific combinations 
of the countermeasures have a different influence on the spreading process. 
For example, with $R_0=1.6$ social distancing and travel restrictions did 
not really seem to help, while vaccination limited the number of new symptomatic cases per 10,000 persons from 
$\sim100$ to $\sim1$. With $R_0=2.1$, such a significant impact could only be achieved with the 
combination of vaccination, school closure, social distancing and travel restrictions. 
In \cite{Liu08} Liu et al. examined the influence of two parameters, the decay rate of the disease and the 
range a person can be infected in, on the spread of an epidemic in 
the urban environment of the Haizhu district of Guangzhou. The results imply the importance of both parameters. Especially 
the results of the distance parameter, which is influenced by peoples behavior, imply significant 
impact on the disease spreading if manipulated wisely.

Concerning mobility in a 2D field, Valler et al.~\cite{Valler11} analyzed the epidemic threshold for a mobile ad-hoc network. They showed that 
if the connections between devices is given by a sequence of matrices $A_1, \dots A_T$, then for $\lambda_S 
< 1$ no epidemic outbreak occurs, with high probability, where $\lambda_S$ is the first eigenvalue 
of $\Pi_{i=1}^T (1-\delta)I +\beta A_i$ with $\beta$ and $\delta$ being the virus transmission probability and the 
virus death probability, respectively. They also approximated the epidemic threshold for different mobility models 
in a predefined 2D area, such as random walk, Levy flight, and random waypoint.

However, realistic scenarios do not only consider the spreading process itself. In reality, we are 
influenced by many factors, e.g.~the awareness about an epidemic.
In \cite{FGWJ09}, Funk et al.~analyzed the spread of awareness on epidemic outbreaks. 
That is, the information about a disease is also spread 
in the network, and it has its own dynamic. In \cite{FGWJ09} the authors described 
the two spreading scenarios (awareness vs.~disease) by the following model. 
Each individual has a level of awareness, which depends on the number of hops 
the information has passed before arriving to this individual. This was combined 
with the traditional SIR model. It has been shown that in a well mixed 
population, the spread of awareness can result in a slower size of outbreak, however,
it will not affect the epidemic threshold. Nevertheless, if the spread of information 
about a disease 
is considered as a local effect in the proximity of an outbreak, then awareness can 
completely stop the epidemic. The impact of spreading awareness is even amplified if
the social network of infections and informations overlap.
\subsection{Our Results}\label{results}
%
In our dynamic model we integrate the results of \cite{EGKM04}, and assume that every 
individual chooses a location independently and uniformly at random according to the 
power law degree distribution of the corresponding places. 
This is the first analytical result on the spread of epidemics in a dynamic scenario, 
where the impact of the power law distribution 
describing the attractiveness of different locations in an urban area
is
considered.

First we show that in the \emph{emerging country model} it is very unlikely for a (deadly) 
epidemic to wipe out the whole population. 
This holds due to the decreasing number of survivors and infected people over time. 
That is, the (infected and uninfected) population size is decreasing over time while 
the available space for each individual is not. This implies a decreasing probability for 
two individuals to meet, since the available space is large enough for the healthy individuals 
to avoid the infected ones (cf.~Section \ref{sec_theorem1}).
This provides analytical evidence in our model for a conjecture expressed in a historical documentation about the plague in the mid ages \cite{Plague05}.

Second we show that in the \emph{industrialized country model} the use of news services 
combined with an appropriate health care budget limits the number of infected persons 
to a negligible fraction.
Due to the warnings, 
the population will act more carefully and the 
corresponding power law exponent increases. That is, a noticeable number of people avoid locations 
with plenty of individuals. Simultaneously, the number of accommodated persons in these locations 
decreases, as well as the probability for the remaining visitors to become infected.

The model of this paper seems to be completely different from most of the models considered 
in Subsection \ref{related}. In many of these papers, the authors used the well known mean field theory
to model the spread of epidemics in urban areas (e.g.~census areas around airports \cite{VespH1N109}) or 
within different mixing groups \cite{Germann06}. Then, infected individuals may pass the disease to 
every other member of their community, with a certain probability. In our model, the infection is only 
transmitted between individuals being in the same cell at a certain time step, which implies 
completely different results depending on the distribution of individuals among the cells. 
Concerning the geographic mobility models (e.g.~\cite{Valler11}), there are several interesting 
results w.r.t.~different mobility patterns for the individuals, such as Levy flight or random waypoint.
In all these cases, free movement in a 2D simulation field is considered, where the spacial distribution 
of the random walk or Levy flight mobility models are uniform, and the distribution of nodes in the 
neighborhood of an infected device is Gaussian. In our paper, we try to 
integrate the distribution of the attractiveness of different locations in an urban area, which 
seems to follow a power law distribution \cite{EGKM04}.

The main disadvantage of our model 
is that we do not take the personal preferences of different people into account, and we ignore 
any dependency between certain choices (e.g., a married couple is in many cases together at the 
same place). 
Nevertheless,
since most of the above dependencies in real world are positively correlated\footnote{That is, 
persons who meet at some place will more likely meet 
at (other) places too. On the other side, it is not very likely that 
persons who have not met within a certain time frame, 
will meet each other afterwards.},
the results should even be stronger 
in more realistic scenarios. This also seem to hold for periodic mobility \cite{Valler11}. 
It would also be interesting to integrate the Levy flight model to some extent, although this 
in general does not hold in urban environments \cite{SubPop09}.
Nevertheless, even the simple mobility model analyzed in this paper provides evidence for the positive impact of 
non-pharmaceutical interventions (e.g. school closures) and cautious behavior (public gathering bans as well as 
public warnings), which has already been observed in real world 
studies \cite{Markel07}. 

\section{Model and annotation}\label{model_annotation}
In this section we present the formal model used in this paper. 

\paragraph*{Modeling the environment} To model the environment we use a static \textit{grid structure} 
of size $\sqrt{\c n} \times \sqrt{\c n}$, where $\c$ is a constant. That is, our grid contains $\c n$ so 
called cells. The cells represent physical locations an individual can visit, e.g., a restaurant, 
the office, or a concert hall. Each cell may contain nodes (also called individuals), depending on its so called  
\textit{attractiveness}.
In reality there are different places with varying attractiveness in an urban area \cite{EGKM04}.
The attractiveness $d$ of a cell $v$ is chosen randomly with probability proportional to $1/d^\alpha$, where $\alpha$ is a constant larger than $2$ (according to \cite{EGKM04}, $\alpha \approx 2.8$ for locations in Philadelphia). We bound the highest attractiveness to $\sqrt[\alpha]{\c n}$. Since the objects move randomly, the expected number of nodes contained in a specific cell with attractiveness $d$ is given by $n \cdot d / \sum_{i=2}^{\sqrt[\alpha]{\c n}} c \frac{\c n}{i^\alpha}i$, where $c$ is some normalizing constant such that $\sum_{i=2}^{\sqrt[\alpha]{\c n}} c \frac{\c n}{i^\alpha}  =\c n$.
This scenario describes e.g.~a city with different 
locations, and individuals visiting these locations according 
to their attractiveness.  

\paragraph*{Modeling moving individuals} In the grid structure described above there are 
$n$ \textit{nodes} (or individuals), which move from one cell to another in each step. 
Each node chooses the 
target location in some step independently and with probability proportional to its attractiveness. Now, assume that 
an infection starts to spread among the nodes. To model the spreading process,
we use three different states, which partition the set of nodes in three groups; $\I(j)$ contains the infected nodes in step $j$, $\U(j)$ contains the uninfected (susceptible) nodes in step $j$ (i.e. subject to infection but not infected already) and $\R(j)$ contains the resistant nodes in step $j$ (i.e., nodes which became cured and can not be infected anymore). Whenever it is clear from the context, we simply write $\I$, $\U$, and $\R$, respectively. If at some step $j$, an uninfected node $i$ visits a cell which also contains a node of $\I(j)$, then $i$ becomes infected and carries the disease further\footnote{This model 
can easily be extended to the case, in which the disease is transmitted according to a given probability.}.  

\paragraph*{Possible variations} In our model, there are only three variables which can vary. One is $\alpha$, which represents the 
probability distribution of the attractiveness of different locations. Another is $\c$, which describes the 
total number of cells, i.e., the locations the nodes may visit. In addition, 
the time until an
infected individual is cured again can also vary. That is, a time period $\tau$ is assigned to the epidemic, which means that an individual
is infective for $\tau$ consecutive steps. With $\tau$ very large (i.e., $\omega(\log n)$) we obtain
the model without any recovery.

\paragraph*{Parameters}
In the emerging country model we assume $2 < \alpha < 3, \c = 1$ and $\tau = \smallO(\log n)$. In contrast, due to the countermeasures applied in the industrialized country model, we assume that $\alpha, \c, \tau$ are large constants there. In both models an uninfected node becomes infected with probability $1$, as soon as an infected node is accommodated in the same cell.
\section{Analysis}\label{analysis}
%
We start our analysis with some basic observations.
\begin{observation}\label{lem_cell_d}	
The expected number of nodes choosing a specific cell with attractiveness $d$ is proportional to $d$.
\end{observation}
\begin{observation}\label{lem_area_d} 
The probability for a node to choose an arbitrary cell with attractiveness $d$ is proportional to ${d^{-\alpha + 1}}$.
\end{observation} 

As we can see, 
the number of cells with a high attractiveness decreases with an increasing $\alpha$.
On the other side, 
while $\c$ increases, the area, in which infected nodes may infect other nodes, decreases. Then,
the probability for two nodes
to choose the same cell decreases too. 

\subsection{Emerging country model}\label{sec_theorem1}
One major concern about the breakout of a deadly disease
is that the epidemic may wipe out the majority of the population. 
At the beginning of the outbreak people become infected. Then, they carry the disease and distribute it among 
the individuals they meet at different locations. 
Let us assume a non curable 
course of disease, in which the infected individuals decease after some time period. 
At some time, there is only a small fraction of the population which is still alive, and some 
of them still carry the infection further. We model this situation 
(at some time step $j$) assuming that $|\I(j)| \cdot |\U(j)| 
\leq n^{2 \epsilon}$, for some small constant $\epsilon$.
The most interesting question is whether these nodes manage 
to infect the remaining healthy nodes, and exterminate the whole population.

\begin{lemma}\label{lemma1}
Let $|\I(j)| \cdot |\U(j)| \leq n^{2\epsilon}$, where $\epsilon$ is an arbitrarily small constant, $\c$ is a constant and $\tau$ is a slow growing function 
in $n$ (i.e., $\tau = \smallO (\log n)$). Then in a round there is no newly infected node with probability $1- n^{-\Omega(1)}$.
\end{lemma}
\begin{proof}
Let a pairing of two nodes $i \in \I(j)$, $i' \in \U(j)$ describe the event that $i$ and $i'$ choose the same cell.
Let $x_{i,i',d}$ be the event that nodes $i$ and $i'$, where $i\neq i'$, choose the same (specific) cell with attractiveness $d$. Then, the probability $Pr(x_{i,i',d})$ is bounded by 
	\begin{align}
		Pr(x_{i,i',d}) &= \left( \frac{d}{\sum_{i=2}^{\sqrt[\alpha]{n}} c \frac{n}{i^\alpha}i} \right)^2 \nt
			\leq \left( \frac{d}{c n \frac{1}{\alpha-2} \left( \frac{1}{2^{\alpha-2}} - \smallO(1) \right)} \right)^2 \nt 
			\leq \left( \frac{d(\alpha-2)2^{3-\alpha}}{c n} \right)^2. \nt
	\end{align}
Further, let $x_{i,i'}$ be the event that node $i$ and $i'$ meet in an arbitrary cell. Then,  
$Pr(x_{i,i'}) \leq \sum\limits_{d=2}^{\sqrt[\alpha]{n}} c \frac{\c n}{d^{\alpha} }
Pr(x_{i,i',d})$, and we obtain
	\begin{align}
		Pr(x_{i,i'}) &\leq \left( \frac{(\alpha - 2) 2^{3-\alpha}}{c} \right)^2 \sum\limits_{i=2}^{\sqrt[\alpha]{n}} \frac{c \c n}{i^\alpha} \frac{i}{n} \frac{\sqrt[\alpha]{n}}{n} 
		\leq \frac{\c (\alpha - 2) \left( 2^{3-\alpha} \right)^2}{ c n^{1-1/\alpha} } (1+\smallO(1)). \nt
	\end{align}

Now, a fixed uninfected node $i'$ becomes infected with probability at most 
$|\I(j)|Pr(x_{i,i'})$ and the expected number of newly infected nodes is bounded by 
	\begin{align}
		\mu &\leq |\U(j)| |\I(j)| Pr(x_{i,i'}) \nt 
			\leq \frac{\c (\alpha - 2) \left( 2^{3-\alpha} \right)^2}{ c n^{1-1/\alpha-2\epsilon} } (1+\smallO(1)). \nt
	\end{align}
Since the nodes of $\U(j)$ are assigned to the cells independently,
we use Chernoff bounds \cite{Che52} to obtain the desired result. With $(1+\delta)\mu=1$ and $X$ being the random variable describing the number of newly infected nodes, we obtain
	\begin{align}
		Pr[X \geq (1+\delta)\mu] &\leq \left[ \frac{e^\delta}{(1+\delta)^{(1+\delta)}} \right]^\mu 
		\leq \left( \frac{e^{\nicefrac{1}{\mu}-1}}{(\nicefrac{1}{\mu})^{\nicefrac{1}{\mu}}} \right)^\mu = \mu \cdot \frac{e}{e^\mu} = n^{-\Omega(1)}. \nt
	\end{align}
\end{proof}
\vspace{10pt}
In the next theorem, we show that at least a polynomial fraction of the population remains uninfected, 
even if no countermeasures are taken.

\begin{theorem}\label{theorem1}
Let $\c = 1$ and let $\tau$ be a slow growing function in $n$ (i.e., $\tau = \smallO (\log n)$). Then, a polynomial fraction of 
the population remains uninfected, when the spreading process runs out.
\end{theorem}
\begin{proof}
We analyze the procedure in three phases. The first phase contains only the 
increase in the number of infected nodes. However, we will ensure that a sufficient number of uninfected 
nodes will still be present, where this number may be very small compared to the size of the population. 
In the second phase we show that either the spreading process runs out at some time when 
the number of uninfected nodes is polynomial in $n$, or we will have a situation where the 
assumptions of Lemma \ref{lemma1} are fulfilled.
Consequently, we apply Lemma \ref{lemma1} in the third phase. 
Let $i$ denote the current time step, where it holds $i=0$ at the beginning, 
and let $t_1, t_2$ describe 
the step \textit{after} the corresponding phases 1 and 2 have ended.

\paragraph{Phase 1}
For this phase we assume, that no node becomes cured. Thus, all infected nodes remain 
infected during this phase, and carry the disease to all nodes they meet in the cells.
According to the power law distribution of the attractiveness, a constant fraction of 
cells will have attractiveness $d=2$. Remember that $\c = 1$, and hence the number of all cells 
is $n$ in total\footnote{The result can easily be extended to 
arbitrary constant $\c$.}. The number of cells not hosting any infected node can be modeled by a simple 
balls into bins game, and we obtain that the number of such cells is $\Theta(n)$ \cite{RS98}.
Then, the probability that an uninfected node remains uninfected in one step is 
also a constant. Since the uninfected nodes are assigned to the cells independently, 
we apply Chernoff bounds to conclude that as long as $|\U(i)| = n^{\Theta(1)}$, at least 
a constant fraction of $\U(i)$ remains uninfected after step $i$, with probability $1-n^{-\Omega(1)}$. 

\paragraph{Phase 2}
We now add the curing procedure to our analysis. Thus, $|\R(t_1+i)|$ will increase in each step. 
However, we know from the first phase that $|\U(j+1)| = \Theta(|\U(j)|)$ for any $j$. Then, the 
number of uninfected nodes will not decrease below some polynomial in $n$, or there is some step 
$j$ such that $|\I(j) \cup \R(j)| = n - n^{\epsilon'}$ and $|\U(j)| = n^{\epsilon'}$ for some 
$\epsilon' >0$ small enough. Assuming that $\tau =o(\log n)$, after $\tau$ steps it holds that 
$\R(i+\tau) = \I(i) \cup \R(i)$ and $|\U(i+\tau)| = |\U(i)|/n^{o(1)}$. Thus,
$$|\I(i+\tau)| \cdot |\U(i+ \tau)| \in \left( \frac{n^{\epsilon'}}{n^{o(1)}} , n^{2 \epsilon'} \right) ,$$
where we assumed that $\I(i+\tau) \neq \emptyset$. Thus, there is some constant $\epsilon>0$ small enough,
such that $|\I(i+\tau)| \cdot |\U(i+ \tau)| = n^{2 \epsilon}$.  

\paragraph{Phase 3}
Since at this point $|\I(t_2)| \cdot |\U(t_2)| \leq n^{2\epsilon}$, we apply Lemma \ref{lemma1} for $\tau$ steps, and obtain the theorem.
\end{proof}
\vspace{10pt}
Theorem \ref{theorem1} implies that if $\tau = o(\log n)$, the disease will run out, and a polynomial 
fraction of the population survives. This is surprising, given the fact that an aggressive virus 
with a long time frame for transmitting the infection (i.e., $\tau$ unbounded) is spread among the 
individuals of a population. Theorem \ref{theorem1} seems to explain the behavior of 
certain epidemics known from the history. We know that some of these epidemics 
exterminated a large fraction of the population of various cities. However, a fraction of the 
citizens could always survive (e.g. the plague \cite{Plague05}).

\subsection{Industrialized country model}\label{sec_theorem2}
In reality the most effective countermeasures against the spread of diseases 
(except vaccination) seem to be \textit{warnings} and \textit{isolation} \cite{Markel07}. 
When a new disease breaks out, the media warns the public and informs it 
about the risks. The most important facts are mentioned, including the ways of transmission 
and possible precautions. A few examples are the bird flu, swine 
flu, and SARS. In reality it is often not possible to isolate infected people at the  
time they become infected. Thus, such an infected individual may potentially infect others until 
his symptoms show up and he becomes isolated. We assume that such a disease has an incubation 
time not bigger than a constant $\s$, and model this fact by setting $\tau= \s$. Now it remains 
to include warnings spread through the media 
into our model. A warning basically affects 
the constants $\alpha$ and $\c$, since the individuals will most likely avoid 
places with a large number of persons, waive needless tours, and be more careful 
when meeting other people. 
Although these modifications alone are most likely not sufficient, we show that a combination of these strategies,
which are able to sufficiently influence the constants $\alpha$ and $\c$, 
are enough to stop the spread of the disease in our model.
Therefore, we assume in our analysis that $\alpha$ and $\c$ are large constants.

\begin{lemma}\label{lem_mart_theorem2}
Let the set $G_k$ contain all cells with attractiveness $2^k$ up to $2^{k+1} - 1$ and let 
$|\I(j)| = f^q(n)$ for a specific step $j$, where $f(n)$ is some function with $\lim_{n\rightarrow 
\infty} f(n) = \infty$ and $q > 3$ constant. 
If $k \leq 
\frac{1}{(\alpha-2)} \cdot \log \left( \frac{|\I(j)|}{f^3(n)} \right)$, then 
the number of newly infected nodes in $G_k$ is bounded by $\bigO\left( \frac{|\I(j)|}{(2^k-1)^{\alpha-3}} \right)$ 
with probability $1-\frac{1}{e^{\Omega(f(n))}}$.
\end{lemma}
\begin{proof}
According to our assumption, the expected number of infected nodes at the beginning of 
step $j$, which choose cells in group $G_k$ with 
$k \leq 
\frac{1}{(\alpha-2)} \cdot \log \left( \frac{|\I(j)|}{f^3(n)} \right)$, is 
$$\sum_{d=2^k}^{2^{k+1} - 1} \frac{|\I(j)| d}{\sum_{i=2}^{\sqrt[\alpha]{\c n}} c \frac{\c n}{i^\alpha}i}
= \Theta\left( \frac{|\I(j)|}{(2^k-1)^{\alpha-2}} \right) = \Omega (f^3(n)) .
$$
To prove the lemma, we formulate the problem as a vertex exposure martingale \cite{MR07}. 
Let $z_{a,b}$ be the event that node $a$ and $b$ choose the same cell. We 
define a graph $G=(V, E)$ by setting $V=\I(j) \cup L$, where $L=\{ 1, ..., n - |\I(j)| \}$ 
and $\I(j) \cap L = \emptyset$. The set of edges is $E=\{ (x_i, l) \mid x_i \in \I(j) \wedge l \in 
L \wedge z_{x_i, l} \}$. Let now be the vertex exposure sequence given by $x_1,...,x_{\I(j)}
\in \I(j)$. 
Thus, each $x_i$ represents an infected node which may establish edges connecting $x_i$ to the set $L$. 
By standard Chernoff bounds \cite{Che52,HR90} it follows that if $\alpha$ is large enough, i.e., 
$2^k = 2^{\nicefrac{1}{(\alpha-2)} \cdot \log \left( \nicefrac{\I(j)}{f^3(n)} \right)} =\smallO(f(n))$, then
an infected node has $\bigO(f(n))$ edges with probability 
$1- \nicefrac{1}{e^{\Omega(f(n))}}$ 
By using the union bound, we may conclude that with probability $1- \frac{f^q(n)}{e^{\Omega(f(n))}}$ 
all infected nodes placed in cells of the group $G_k$ have $\bigO(f(n))$ edges.

Now we restrict the probability space to the events, in which  
$\bigO\left( \frac{|\I(j)|}{(2^k-1)^{\alpha-2}} \right)$ infected nodes 
choose cells in $G_k$, and each of these nodes has $\bigO(f(n))$ edges.
We obtain such an event with probability $1-e^{-\Omega(f(n)}$.
Let $X_0, X_1, ...$ be the vertex exposure martingale, where $X_i$ is the 
expectation on the number of edges incident to $V \cap \I_{G_k}(j)$, conditioned by the knowledge of the edges 
incident to $x_1, \dots , x_i$, where $\I_{G_k}(j)$ is the set of infected nodes choosing cells in $G_k$.
%
Note that we only expose the infected vertices lying in $G_k$.
Given the assumption above, 
for any $k\leq 
\frac{1}{(\alpha-2)} \cdot \log \left( \frac{|\I(j)|}{f^3(n)} \right)$ 
we have $|X_k - E[X_k]| \leq \bigO(f(n))$. Then,
\begin{align*}
&Pr(|X_{|\I_{G_k}(j)|}| - X_0| \geq \lambda) \leq 2\cdot e^{-\frac{\lambda2}{2\sum_{k=1}^{t} f(n)2}} 
\leq \exp{\left( -\frac{1}{\c} \frac{f^3(n)}{\bigO( f^2(n))} \right) },
\end{align*}
with $\lambda = \frac{|\I(j)|}{\c(2^k -1)^{\alpha-2}} \geq \frac{f^3(n)}{\c}$.
Since the expected 
number of uninfected nodes in a cell of $G_k$ is less than $2^{k+1}-1$, we obtain $X_0 = 
\bigO\left(\frac{|\I(j)|}{(2^k-1)^{\alpha-2}}\right) \cdot 
(2^{k+1}-1) = \bigO\left(\frac{|\I(j)|}{(2^k-1)^{\alpha-3}}\right)$, and the lemma follows.
\end{proof}

\begin{theorem}\label{theorem2}
Let an epidemic disease be spread in a $\sqrt{\c n} \times \sqrt{\c n}$ network
as described in section \ref{model_annotation}. Furthermore, let 
$\tau = \s$, where $\c$, $\alpha$, and $\s$ are some large constants. Then, the 
network is healthy again after $\bigO((\log\log n)^4)$ steps, 
with probability $1-\smallO(1)$. Moreover, when $\I= \emptyset$, 
the set $\R$ has size $\log^{O(1)} n$.  
\end{theorem}
\begin{proof}
The proof consists of three parts. In the first part we show that the number of infected 
nodes decreases after $\s$ steps by at least a constant factor, 
with probability 
$1-2\s |\I(j)|^{-1/q}$, where $q$ is a constant. 
The second part states a result about the oscillating behavior (w.r.t.~the number 
of infected nodes) during the whole process. Finally the third part shows, that 
if the number of informed nodes is just small enough 
(although dependent on $n$), it is sufficient to consider $\bigO((\log \log n)^4)$ additional 
steps to eliminate the remaining infected nodes, 
with probability $1- \log^{-\Omega(1)}(n)$. 
Then, all parts together imply the validity of the theorem.

\paragraph{Part 1}
Let $\nu$ be a proper upper bound on the number of newly infected nodes in the network. 
This bound holds with probability $1-\bigO(|\I(j)|^{-1/q})$ at step $j$ and will be computed later. 
By setting $\tau = \s$ the infected nodes will carry the infection further 
only for a constant time $\s$. Let a super-step be a sequence of $\s$ consecutive steps. We know that
after a super-step the nodes which were infected at the beginning of this super-step, 
become cured (i.e., they are moved to the set $\R$). 
Let $j \geq \s$ be an arbitrary time step.
Then,
	\begin{align}
		|\I(j)| &\leq |\I(j-\s)|(1+\nu)^\s - |\I(j-\s)| 
		       \leq |\I(j-\s)| \left( (1+\nu)^\s - 1 \right). \EQlabel{1}
	\end{align}
	
To obtain a value for $\nu$ we ask 
how many different cells become occupied by the infected nodes at most. 
Therefore, we group the cells with respect to their attractiveness.
The group $G_k$ contains all cells with attractiveness 
$2^k$ up to $2^{k+1} - 1$. Let $x_{ij}$ describe the number of nodes, which are infected by 
$i \in \I(j)$ in step $j$, and define $X_j:=\sum_{\forall i \in \I(j)} x_{ij}$.
Let us first assume that $|\I(j)| = \log^q n$. 
Note that if $\alpha$ is large enough (cf.~Lemma \ref{lem_mart_theorem2}), the expected 
number of infected nodes contained in all the cells of attractiveness at least 
$2^{\frac{1}{\alpha - 2} \log\left( \frac{|\I(j)|}{\log^3 n} \right)}$ 
is bounded by $\bigO(\log^3 n)$, with probability $1-\rho^{-\log(n)}$, where $\rho$ is some proper 
constant. 
If the attractiveness is larger than $l = 2^{\frac{1}{\alpha - 2}\log\left( |\I(j)| 
\log^{1+\Theta(1)} n \right)}$, then the expected number of infected nodes 
in all these cells together is $\smallO(\log^{-3} n)$.

Let $E_{k,n}$ be the expected number of infected nodes 
contained in $G_k$. Note that $E_{k,n}$ decreases with the number of infected 
nodes. According to Lemma \ref{lem_mart_theorem2}, if $f(n) = \log n$, then 
the number of newly infected 
nodes in $G_k$ in some step $j$ is $\bigO\left( \frac{|\I(j)|}{(2^k-1)^{\alpha-3}} \right)$, with probability 
$1-n^{-\Omega(1)}$. Thus, 
we obtain 
	\begin{align}
		 X_j &\leq \sum\limits_{k=2}^{\frac{1}{\alpha -2} \log\left( \frac{|\I(j)|}{\log^3 n} \right)} \bigO \left( \frac{|\I(j)|}{(2^k -1)^{\alpha -3}} \right)
				 + \bigO(\log^3 n) \cdot \left( 2 \cdot l -1 \right) \nt\\
			 &\leq {c'} \cdot |\I(j)| \left( 1+ \frac{\bigO( \log^{1+\Theta(1)}n )}{|\I(j)|^{1-\frac{1}{\alpha-2}}} \right)
			 = {c'} \cdot |\I(j)| \left( 1+ \smallO(1) \right), \EQlabel{2}
	\end{align}
with probability $1-\smallO(\log^{-1} n)$. Note that considering the attractiveness up to $l$ only is crucial here, since 
no infected node is in some cell with attractiveness larger than $l$ with probability $1-\smallO(\log^{-1} n)$. 
In this context $c'$ is  a small constant depending on 
$\alpha$ and $\c$.
For decreasing $|\I(j)|$ we obtain the general formulas where $\log(n)$ is replaced by $f(n)$ in the statements above. 
Note that at this point one can represent $|\I(j)|$ as $f^q(n)$. Then 
inequality (\EQref{2}) holds with 
probability $1-|\I(j)|^{-1/q}$. Thus, $\nu = {c'}$ with the probability given above.

Then, it holds that
\begin{align}
		\left( (1+\nu)^\s - 1 \right) = \sum\limits_{k=0}^{\s} \left( \binom{\s}{k} 1^{\s-k} \nu^k \right) - 1 = \sum\limits_{k=1}^{\s} \binom{\s}{k} 1^{\s-k} \nu^k
		\leq \frac{(\nu \s) \frac{(1-\nu)^{\s-1}}{(1-1/\s)^{\s-1}}}{(1-\nu)^{\s-1}} 
		\leq \nu \s e^{\frac{\s-1}{\s}} << 1, \EQlabel{3}
\end{align}
if $c'\s << 1/e$. Here, the first inequality in (\EQref{3}) holds using the following estimation from \cite{HR90}
	\begin{align}
	 \sum\limits_{k=1}^{\s} \binom{\s}{k} 1^{\s-k} \nu^k (1-\nu)^{\s-1} 
		\leq \left( \frac{\nu}{1/\s} \right)^{\frac{1}{\s}\s} \left( \frac{1-\nu}{1-1/\s} \right)^{(1-1/\s)\s}. \nt
	\end{align}
Then (\EQref{1}), (\EQref{2}) and (\EQref{3}) imply a decreasing number of infected nodes after each super-step by a factor of $\nu \s e$, with probability at least $1-2\s |\I(j)|^{-1/q}$.

\paragraph{Part 2}
The results so far imply an oscillating behavior of the number of infected nodes. That is, 
the set $\I$ will mainly decrease (cf.~Part 1). However, the number of infected nodes 
may also increase in some time steps. If one infected node chooses a destination with plenty of nodes from 
$\U$, then $\I$ increases drastically. Thus, it remains to compute the probability for such an 
event. 

Now we divide the whole process into {\em phases}. Lets consider each phase separately, and 
assume that a specific phase begins at time step $j$. A phase consists usually of $\s$ steps, and a step is 
divided into two substeps. In the first substep of a step $j'$, we allow each infected node to transmit the 
disease to every other node being in the same cell. In the second substep, all nodes 
which were infected in step $j'-\s$ are moved to $\R$, and stop transmitting the disease in 
the subsequent steps. 

Now, a phase starting at some time step $j$ ends after $\s$ steps, if in all these steps 
$j' \leq \s$ it holds that $|\I(j+j') \setminus \I(j)| \leq \nu \s e |\I(j)|$. This holds with 
probability at least $1-2\s |\I(j)|^{-1/q}$ (see above).
If in some step $j' \leq \s$
$|\I(j+j') \setminus \I(j)| > \nu \s e |\I(j)|$, then this phase ends, and in the next step we start 
with a new phase. 

Now we model the process as a special random walk (cf. Figure \ref{fig_rw}). For this, let $G'=(V',E')$ be a directed graph,
where $V'= \{ \log \log \log \log n , \dots , \log^{q} n\}$.
A node $v \in V'$ corresponds to the case $|\I| = v$. From each $v$, there is a transition to $\max\{v_{\min}, \nu \s e v\}$ 
with probability $1-2\s v^{-1/q}$, where 
$v_{\min} = \log \log \log \log n$ (all the other transitions are not relevant and can be arbitrary).
In order to show that within $\bigO ((\log \log n)^4)$ steps the vertex $v_{\min}$ 
is visited $(\log \log n)^2$ steps, we state the following lemma.
\begin{lemma}
\label{hilfslemmatheo}
Let $(X_t)_{t=1}^{\infty}$ be a Markov chain with space $\{1, \dots , m\}$ that fulfills the 
following property:
\begin{itemize}
\item there are constants $c_1,c_2 <1$ such that for any $t \in \N$,
$Pr[X_{t+1} \leq c_1 X_t] \geq 1-X_t^{-c_2}.$
\end{itemize}
Let $T= \min\{t \in \N ~|~ X_t =1\}$. Then, 
$$Pr[T =\bigO(\log^2 m)] \geq 1-m^{-4}.$$
\end{lemma}
\begin{proof}
The proof follows from some of the arguments of Claim 2.9 from \cite{DGMSS11}. However, there 
are two main differences. First, the probability of not moving toward the target vertex (which is $m$ 
in their case and $1$ in our case)
is in \cite{DGMSS11} exponentially small w.r.t.~the current state $-X_t$, while we only have some polynomial 
probability in $X_t^{-1}$. Second, in their case the failure in moving toward the 
traget vertex most likely occurs at the other end of the graph (close to $1$), 
while in our case a failure mainly occurs at some state, which is close to 
the target. 

We define a step to be decreasing, if $X_{t+1} \leq c_1 X_t$. Furthermore, a decreasing 
step is successful, if $X_{t+1} =1$ after some $X_t>1$. Let $Y$ be a random variable which 
denotes the number of consecutive decreasing steps without reaching a successful step.
Note that $t=\log_{1/c_1} m $ is an upper bound on $Y$ since 
$X_{t} \leq c_1^{\log_{1/c_1} m} m =1$. 

We divide now the Markovian process $X_t$ into consecutive epochs. Every time we 
fail to have a decreasing step, we start a new epoch. Furthermore, if a step is 
successful, we stop. Now we show that an epoch is successful with 
some constant probability. Let $P_i$ be the probability that 
some epoch $i$ contains a successful step. Then,
\begin{align}
P_i &\geq \prod_{j=0}^{\log_{1/c_1} m-1} \left( 1- (c_1^j m)^{-c_2}\right)
\geq  \prod_{j=0}^{\log_{1/c_1} m-1} \rho^{-(c_1^j m)^{-c_2}}
= \rho^{-a} \nt
\end{align}
where $\rho$ is a proper constant such that $1 < \rho \leq \left( 1- c_1^{c_2} \right)^{-c_1^{-c_2}}$. Then for the exponent $a$ we have
\begin{align}
a = \sum_{j=0}^{\log_{1/c_1} m-1} (c_1^j m)^{-c_2}
= \sum_{j=0}^{\log_{1/c_1} m-1} \frac{(1/c_1^{c_2})^j}{m^{c_2}}
= \frac{1-(1/c_1^{c_2})^{\log_{1/c_1} m}}{(1-1/c_1^{c_2})m^{c_2}}
= \frac{m^{c_2}-1}{(1/c_1^{c_2}-1)m^{c_2}} = \Theta(1). \nt
\end{align}

Let $T'$ be the first epoch in which we reach a successful step. As in \cite{DGMSS11}, we know that an epoch 
has such a step with constant probability, and hence
$Pr[T' = \bigO (\log m)] \geq 1-m^{-4}$. 
Since an epoch can only last for $\bigO (\log m)$ steps, we obtain that a successful 
step is reached within $\bigO (\log^2 m)$ steps, with probability $1-m^{-4}$. 
\end{proof}

According to the lemma above, in the graph $G'$, node $v_{\min}$ is visited at least once 
within $\bigO((\log \log n)^2)$ phases, with probability $1-\log^{-4q} n$. Furthermore, if $|\I(j)| =
\log^{q} n$ for some $j$, then $|\I(j+\s)| =  \nu \s e |\I(j)|$ with probability $1-2\s\log^{-1} n$.
Therefore, using the union bound we conclude that 
$v_{\min}$ is visited $(\log\log n)^2$ times by the random walk within $\bigO((\log \log n)^4)$ phases, without visiting $\log^{q} n$
twice in any two consecutive phases, with probability $1-\smallO(\log^{-0.99} n)$.

\paragraph{Part 3}
\begin{figure}[!t]
\begin{center}
\includegraphics[width=3in]{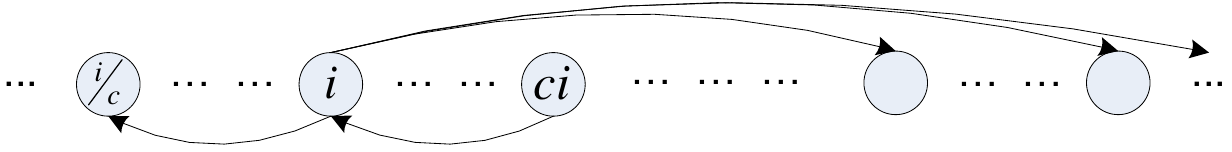}
\caption{Model of the random walk for part 2 of Theorem \ref{theorem2}, where $c=(\nu\s e)^{-1}$. Thereby the edges were partially plotted for the node $i$.}
\label{fig_rw}
\end{center}
\end{figure}

In order to conclude, we observe that within one step, $|\I|$ does not increase to some value larger 
than $\log^q n$ with probability $1-\smallO(\log^{-1} n)$ (cf.~inequality (\EQref{2})). Furthermore, 
given that  $|\I|$ is always smaller than $\log^q n$,
we know that within $\bigO((\log \log  n)^4)$ steps, with probability $1-\smallO(\log^{-0.99} n)$ 
there are $\Theta(\log^2 \log n)$ time steps, in which the number of infected 
nodes is at most $\log \log \log \log n$. To show that the disease becomes 
eliminated from the system, we 
consider the probability for these nodes not to meet any other node for $\tau$ 
consecutive time steps.

To show our claim, we only consider time steps in which the number of infected nodes is at most 
$\log \log \log \log n$. First, we assign the uninfected nodes to the cells. 
As in Lemma \ref{lem_cell_d}, the probability for a node to choose a specific cell with constant attractiveness 
is proportional to $\frac{\bigO(1)}{\c n}$. Thus, at least a constant fraction  
of the cells will remain empty, since the probability for a cell with constant attractiveness to
be empty is $\left( 1- \nicefrac{\bigO(1)}{\c n} \right)^n = e^{-\Theta(1)}$.
Let $t=\log^2\log n$ be the number of so called phases, 
where each phase consists of $\tau$ steps. Then, an (infected) node chooses an empty cell for
$\tau$ consecutive steps with probability $e^{-\Theta(\tau)}$ and all nodes of $\I$ choose 
empty cells in $\tau$ consecutive steps with probability $e^{-\Theta(\tau |\I(j)|)}$. 
Then, in all $t$ phases there is at least one node $v \in \I$ which 
does not choose an empty cell in at least one of the steps with probability
	\begin{align}
		\left( 1- \frac{1}{e^{\Theta(\tau |\I(j)|)}} \right)^t \leq \exp{\left( -\frac{t}{e^{ \Omega\left( |\I(j)| \right)}} \right)} \leq \log^{-\Omega(1)} n. \EQlabel{6}
	\end{align}
Hence, after $\bigO(\log^4\log n)$ steps there is at least one phase, in which all nodes of $\I$ spend $\tau$ consecutive steps alone in some 
cells, with probability $1-\log^{-\Omega(1)} n$.

Thus, the network becomes completely healthy after $\bigO((\log \log  n)^4)$ 
steps with probability $1-\smallO(1)$, and the theorem follows.
\end{proof}

\section{Conclusion}\label{conclusion}
We presented a model, which describes a simple movement behavior of individuals as well as the impact of certain countermeasures 
on the spread of epidemics in an urban environment. Two different parameter settings were used for the analysis. In the first case the epidemic can spread nearly unhindered. The obtained result for this case implies that w.r.t. our model a part of the population will survive with probability $1-\smallO(1)$.
In the second case the epidemic is combated by public warnings, isolation and (limited) medications. One can observe, that in this case the epidemic is embanked after a short time with probability $1-\smallO(1)$. Furthermore the number of total infections decreases from a high percentage of the population to a negligible fraction.

Nevertheless, several open questions remain. In our model, we assumed that every node chooses a cell with attractiveness $d$ with probability proportional do ${d^{-\alpha+1}}$. However, different individuals may have different preferences which are not included in our analysis. Furthermore, different types of movement models are conceivable like Levy flight, periodic mobility model, and grid like movement. 
Although all these characteristics are not considered in this paper, our methods and techniques might be useful to analyze more realistic movement models in the future.


\newpage
\bibliographystyle{splncs03}
\bibliography{Bibliography}

\end{document}